\documentclass[11pt]{article}
\usepackage{amsmath,amssymb,amsthm,amsfonts,latexsym,xspace,enumerate,graphicx,color,eucal,upgreek,mathtools}
\usepackage[backref,colorlinks,bookmarks=true]{hyperref}

\newtheorem{theorem}{Theorem}[section]
\newtheorem*{namedtheorem}{\theoremname}
\newcommand{\theoremname}{testing}

\newtheorem{lemma}[theorem]{Lemma}

\newtheorem{claim}[theorem]{Claim}

\newtheorem{corollary}[theorem]{Corollary}

\newtheorem*{question*}{Question}
\newtheorem*{metaquestion*}{Meta Question}
\newtheorem*{maintheorem*}{Theorem}

\theoremstyle{definition}
\newtheorem{definition}[theorem]{Definition}

\theoremstyle{plain}
\newtheorem{Alg}{Algorithm}

\renewenvironment{proof}{\noindent{\textbf{Proof:}}} {$\blacksquare$\vskip \belowdisplayskip}

\newcommand{\ignore}[1]{}

\newcommand{\calC}{{\mathcal{C}}}

\newcommand{\calS}{\mathcal{S}}
\newcommand{\calT}{\mathcal{T}}

\begin{document}

\title{A Singly-Exponential Time Algorithm for Computing Nonnegative Rank}
\author{Ankur Moitra\thanks{Institute for Advanced Study, \texttt{moitra@ias.edu}. Research supported in part 
by NSF grant No.
DMS-0835373 and by an NSF Computing and Innovation Fellowship.} }

\maketitle

\begin{abstract}
%Here we give an algorithm for deciding if the nonnegative rank of $M$ is at most $r$ that runs in time $O((nm)^{r^2})$ if $M$ is dimension $m \times n$. Hence, this algorithm runs in polynomial time for any constant $r$ and is the first exact algorithm that runs in time singly-exponential in $r$. Our result is a considerable simplification of the result in Arora et al \cite{AGKM} who give a doubly-exponential time algorithm for this problem that runs in time $O((nm)^{c^r})$. In fact, our result is nearly optimal under the Exponential Time Hypothesis (ETH), given the recent results of Arora et al \cite{AGKM} who demonstrate that an algorithm running in time $O((nm)^{o(r)})$ would yield a sub-exponential time algorithm for $3$-SAT. 

Here, we give an algorithm for deciding if the nonnegative rank of a matrix $M$ of dimension $m \times n$ is at most $r$ which runs in time $(nm)^{O(r^2)}$. This is the first exact algorithm that runs in time singly-exponential in $r$. This algorithm (and earlier algorithms) are built on methods for finding a solution to a system of polynomial inequalities (if one exists). Notably, the best algorithms for this task run in time exponential in the number of variables but polynomial in all of the other parameters (the number of inequalities and the maximum degree). 

Hence these algorithms motivate natural {\em algebraic} questions whose solution have immediate {\em algorithmic} implications: How many variables do we need to represent the decision problem, does $M$ have nonnegative rank at most $r$? A naive formulation uses $nr + mr$ variables and yields an algorithm that is exponential in $n$ and $m$ even for constant $r$. (Arora, Ge, Kannan, Moitra, STOC 2012) \cite{AGKM} recently reduced the number of variables to $2r^2 2^r$, and here we {\em exponentially} reduce the number of variables to $2r^2$ and this yields our main algorithm. In fact, the algorithm that we obtain is nearly-optimal (under the Exponential Time Hypothesis) since an algorithm that runs in time $(nm)^{o(r)}$ would yield a subexponential algorithm for $3$-SAT \cite{AGKM}. 

Our main result is based on establishing a normal form for nonnegative matrix factorization -- which in turn allows us to exploit algebraic dependence among a large collection of linear transformations with variable entries. Additionally, we also demonstrate that nonnegative rank cannot be certified by even a very large submatrix of $M$, and this property also follows from the intuition gained from viewing nonnegative rank through the lens of systems of polynomial inequalities.

\end{abstract}

\setcounter{page}{0} \thispagestyle{empty}
\newpage

\section{Introduction}

\subsection{Background}

The nonnegative rank of a matrix is a fundamental parameter that arises throughout algorithms and complexity and admits many equivalent formulations. In particular, given a nonnegative\footnotemark[1] matrix $M$ of dimension $m \times n$, its nonnegative rank is the smallest $r$ for which:

\begin{itemize}

\item $M$ can be written as the product of nonnegative matrices $A$ and $W$ which have dimension $m \times r$ and $r \times n$ respectively

\item $M$ can be written as the sum of $r$ nonnegative rank one matrices

\item there are $r$ nonnegative vectors $v_1, v_2, ... v_r$ (of length $m$) such that the nonnegative hull of $\{v_1, v_2, ... v_r\}$ contains all columns in $M$

\end{itemize}

%Additionally if $M$ has rank $r$ and $M = UV$ where $U$ and $V$ have dimension $m \times r$ and $r \times n$ respectively but are not necessarily nonnegative, $M$ has nonnegative rank $r$ if and only if there is an invertible linear transformation $T$ such that $UT^{-1}$ and $TV$ are both nonnegative \cite{Vav}. 

\noindent  Throughout this paper, we will denote the nonnegative rank by $rank^+(M)$ and we will refer to a factorization $M = AW$ where $A$ and $W$ are nonnegative and have dimension $m \times r$ and $r \times n$ respectively as a nonnegative matrix factorization of inner-dimension $r$. 

\footnotetext[1]{We will refer to a matrix that is entry-wise nonnegative as a "nonnegative matrix".}

Some of the most compelling applications of nonnegative rank are in machine learning, statistics, combinatorics and communication complexity. In machine learning, the benefit of requiring a matrix factorization $M = AW$ to be nonnegative is that this factorization can then be interpreted {\em probabilistically}. A representative application comes from the domain of topic modeling, where $M$ is chosen to be a so-called "term-by-document matrix": the entry in row $i$, column $j$ is the frequency of occurrence of the $i^{th}$ word in the $j^{th}$ document. And computing a nonnegative matrix factorization of inner-dimension $r$ is akin to finding a collection of $r$ topics (which are each distributions on words) so that each document can be expressed as a convex combination of these $r$ topics. Nonnegative matrix factorization  has found applications throughout machine learning, from topic modeling to information retrieval to image segmentation and collaborative filtering. Even this is far from an exhaustive list. We note that of particular interest in these applications, are instances of this problem in which the target nonnegative rank $r$ is {\em small}. 

In combinatorial optimization, one is often interested in expressing a polytope $P$ as the projection of a higher-dimensional polytope $Q$ which (hopefully) has much fewer facets. The minimum number of facets needed is called the {\em extension complexity} of $P$ and there is a rich body of literature on this subject. Yannakakis established a striking connection between extension complexity and nonnegative rank: Given the polytope $P$, one constructs the "slack matrix": the entry in row $i$, column $j$ is how slack the $i^{th}$ vertex is against the $j^{th}$ constraint. Yannakakis proved that the nonnegative rank of the slack matrix is exactly equal to the extension complexity of $P$ \cite{Yan}. Fiorini et al \cite{FMPTW} recently used this connection and results from quantum communication complexity to prove a remarkable lower bound, that the traveling salesman (TSP) polytope has no polynomial size extended formulation. 

In communication complexity, the famous Log Rank Conjecture of Lovasz and Saks \cite{LOG} asks if the log of the rank of the communication matrix and the deterministic communication complexity are polynomially related. In fact, an equivalent formulation of this problem (that follows from \cite{AUY}) is that the Log Rank Conjecture asks if the log of the rank and the log of the nonnegative rank of a Boolean matrix are polynomially related. Of crucial importance here is that the matrix in question be Boolean. For a general matrix, there is no non-trivial relationship since there are examples in which the rank is three and yet the nonnegative rank is $\Omega(\sqrt{n})$ \cite{FRT}. Also in complexity theory, Nisan used nonnegative rank to prove lower bounds for non-commutative models of computation~\cite{Nisan}.

We note that nonnegative matrix factorization has also been applied to problems in biology, economics and chemometrics to model all sorts of processes, ranging from stimulation in the visual cortex to the dynamics of marriage. In fact, a historical curiosity is that nonnegative rank was first introduced in chemometrics, under the name of {\em self-modeling curve resolution}. 

\subsection{Systems of Polynomial Inequalities}\label{sec:intro2}

The focus of this paper is:

\begin{question*}
What is the complexity of computing the nonnegative rank?
\end{question*}

A priori it is not even clear that there is an algorithm that runs in any finite amount of time. But indeed, Cohen and Rothblum \cite{CR93} observed that the decision question of whether or not $rank^+(M) \leq r$ can be equivalently formulated as a system of $O(mn)$ polynomial inequalities with $mr + nr$ total variables variables: we can treat each entry in $A$ and each entry in $W$ as a variable, and the constraint that this be a valid nonnegative matrix factorization is exactly that $A$ and $W$ be nonnegative and that $M = AW$. The latter is a set of $mn$ degree two constraints. It is easy to see that this system of polynomial inequalities has a solution if and only if $rank^+(M) \leq r$. 

Moreover, whether or not a system of polynomial inequalities has a solution is decidable. This is a quite non-trivial statement. The first algorithm is due to Tarski \cite{Tar}, and there have since been a long line of improvements to this decision procedure. The best known algorithm is due to Renegar \cite{Ren} and the running time of finding a solution to a system of $p$ polynomial inequalities with $k$ variables and maximum degree $D$ is roughly $$\Big ( D p \Big )^{O(k)}$$ 
\noindent So (appealing to decision procedures for a system of polynomial inequalities) there is an algorithm for computing the nonnegative rank of a matrix that runs in a finite amount of time. Note that if the target nonnegative rank $r$ is {\em small} (say, three), this algorithm still runs in time exponential in $m$ and $n$. And the question of whether or not there is a faster algorithm (in particular, one which runs in polynomial time for any constant $r$) was still open. Vavasis proved that nonnegative rank is NP-hard to compute \cite{Vav}, but this only rules out an exact algorithm that runs in time polynomial in $n$, $m$ and $r$ (if $P \neq NP$). 

The crucial observation that the reader should keep in mind throughout this paper is that the main bottleneck in finding a solution to a system of polynomial inequalities is the {\em number of variables}. Renegar's algorithm \cite{Ren} runs in time polynomial in the number of polynomials ($p$) and the maximum degree ($D$), but runs in time exponential in the number of variables ($k$). In a technical sense, the number of variables plays an analogous role to the VC-dimension in learning theory. (This connection can be made explicit by drawing an analogy between the Milnor-Thom and Warren Bounds and the Sauer-Shelah Lemma).

Cohen and Rothblum \cite{CR93} give a reduction from nonnegative rank to finding a solution to a system of polynomial inequalities that has $mr + nr$ variables and a natural goal is to try to use fewer variables in this reduction. Arora et al \cite{AGKM}\footnotemark[2] do exactly this and give a reduction to a system with only $f(r) = 2 r^2 2^r$ variables. This yields an exact algorithm for deciding if $rank^+(M) \leq r$ that runs in time $$\Big ( nm \Big )^{ 2 r^2 2^r}$$ \noindent which is doubly exponential in $r$, but runs in polynomial time algorithm for any fixed $r$. Furthermore Arora et al \cite{AGKM} demonstrate that an exact algorithm for deciding if $rank^+(M) \leq r$ that runs in time $(nm)^{o(r)}$ would yield a sub-exponential time algorithm for $3$-SAT. In summary, there is an exact algorithm for deciding if $rank^+(M) \leq r$ that runs in polynomial time for any $r = O(1)$, and any algorithm must depend (at least) exponentially on $r$. However, the algorithm in \cite{AGKM} runs in time doubly exponential in $r$, and perhaps we could still hope for an algorithm that runs in time singly-exponential in $r$. Here, we give such an algorithm and we do this by reducing the number of variables {\em exponentially} from $ 2 r^2 2^r$ to $2r^2$. 

\footnotetext[2]{We remark that the present author is the last author on the paper \cite{AGKM}. However, the proofs that we present here will be self-contained.}

And perhaps the main message in this paper is that systems of polynomial inequalities with even just a small number of variables can be remarkably expressive! We believe that this theme may find other applications: Perhaps there are other problems for which one would like to design an algorithm based on solving some appropriately chosen system of polynomial inequalities. Then in this case, reducing the number of variables can drastically improve the running time of an algorithm. Indeed, maybe this complexity measure deserves to be studied in its own right:

\begin{metaquestion*}
Given a decision problem, how many variables are needed to encode its answer as a system of polynomial inequalities? 
\end{metaquestion*}

\noindent In particular, we want that the decision problem is a \textbf{YES} instance if and only if the corresponding system of polynomial inequalities has a solution. We note that this question probably makes the most sense and is the most promising in the context of geometric problems. (Indeed, nonnegative rank can be thought of in a purely geometric language and this is the view that will be most useful in our paper). 

\subsection{Our Results}

We now state our main results: Let $M$ be a $m \times n$ nonnegative matrix and let $L$ denote the maximum bit complexity of any coefficient in $M$. We prove

\begin{maintheorem*}
There is a $poly(n, m, L) ( r 4^{r+1} m n)^{c r^2} $ time algorithm for deciding if the nonnegative rank of $M$ is at most $r$. Additionally, given $\delta > 0$ (and if $rank^+(M) \leq r$), the algorithm runs in time $poly(n, m, L, \log \frac{1}{\delta}) ( r 4^{r+1} m n)^{c r^2} $ returns factors $\tilde{A}$ and $\tilde{W}$ that are entry-wise close (within an additive $\delta$) to $A$ and $W$ (respectively) that are a nonnegative matrix factorization of $M$ of inner-dimension at most $r$. Furthermore the entries of $\tilde{A}$ and $\tilde{W}$ have rational coordinates with numerators and denominators bounded in bit length by $O(L ( r 4^{r+1} m n)^{c r^2} + \log \frac{1}{\delta})$. 
\end{maintheorem*}

\noindent This is the first algorithm that runs in singly-exponential time as a function of $r$, and in fact is an {\em exponential} improvement over the previously best known algorithm due to Arora et al \cite{AGKM}. Moreover, notice that the algorithm in \cite{AGKM} is faster than the one in \cite{CR93} only if $r = O(\log n)$ whereas our algorithm is in fact faster for any $r = o(n)$. Our algorithm is nearly optimal (under the Exponential Time Hypothesis), since an exact algorithm that runs in time $(nm)^{o(r)}$ would yield a sub-exponential time algorithm for $3$-SAT \cite{AGKM}.

Our approach is based on two steps. First, we establish a "normal form" for nonnegative matrix factorization. We show that any nonnegative matrix factorization $M = AW$ of inner-dimension $r$ can be placed in a normal form (crucially, without changing the inner-dimension) so that a small subset of entries of $A$ and $W$ uniquely determine all of the remaining entries. More precisely, there are functions $F$ and $G$ (whose behavior only depends on an $r \times r$ submatrix of $A$ and on an $r \times r$ submatrix of $W$ respectively) such that $F$ maps each column of $M$ to the corresponding column of $W$ and $G$ maps each row of $M$ to the corresponding row of $A$. 

These functions $F$ and $G$ can be quite complicated when $A$ or $W$ do not have full column or row rank respectively. In the case that both $A$ and $W$ have full column and row rank, these functions are just linear transformations (see \cite{AGKM}). The difficulty is that when, say, $A$ is rank deficient there are cases in which we need exponential (in $r$) many linear transformations $T_1, T_2, ... T_q$ so that the output of $F$ is always the output of one of these linear transformations applied to a column of $M$. This is precisely the reason that the previous algorithm \cite{AGKM} ran in time doubly exponential in $r$ -- the number of variables is dominated by the number of linear transformations that we need, and in some cases we really do need exponentially many linear transformations to define the function $F$. Our approach to circumvent this problem is to exploit {\em algebraic dependence} among these linear transformations. In particular, our normal form allows us to show that the entries in these linear transformations can be defined as (ratios of) polynomial functions of a much smaller number of shared variables. This is an immediate corollary of our normal form and a simple application of Cramer's Rule. Hence we can reduce the number of variables (in the system of polynomial inequalities) from exponential in $r$ to quadratic in $r$.

We also consider another basic question about the nonnegative rank of a matrix: 

\begin{question*}
Can the nonnegative rank of a matrix $M$ be certified by a small submatrix? 
\end{question*}

Indeed -- in the case of the rank -- a matrix $M$ has rank at least $r$ if and only if there is an $r \times r$ submatrix of $M$ that has rank $r$. This property plays a crucial role in many applications \cite{GV} and it is natural to wonder if the nonnegative rank admits any similar characterization. As another motivation, often we are only given a subset of the entries of the matrix $M$ (for example, in the Netflix problem) and we would like to use these entries to infer properties about $M$. 
%There is a large body of work that demonstrates (under suitable conditions on $M$, namely {\em incoherence}) a small, random set of entries are sufficient to recover the remaining entries exactly provided that $M$ has low rank. Here we demonstrate that the situation is quite different for nonnegative rank:
Yet, the nonnegative rank behaves quite differently than the rank in this regard. 

\begin{maintheorem*}
For any $r \in \mathbb{N}$, there is a $3rn \times 3rn$ nonnegative matrix which has nonnegative rank at least $4r$ and yet for any $< n$ rows, the corresponding submatrix has nonnegative rank at most $3r$. 
\end{maintheorem*}

\noindent So even the submatrices consisting of a constant fraction of the rows in $M$ do not determine the nonnegative rank of $M$ even within a constant factor. This result, too, can be thought of in the language of systems of polynomial inequalities: The basic principle at play is that even though the nonnegative rank can be equivalently characterized by a system of polynomial inequalities with only $2r^2$ variables, there are systems of polynomial inequalities that are together infeasible and yet any strict subset of the constraints is feasible. This is in stark contrast to the case of linear inequalities, for which, if the system is infeasible (and is in dimension $d$) there is a subset of just $d$ linear inequalities that is infeasible (i.e. there is a size $d$ obstruction) \cite{Mat}.

\section{Computing the Nonnegative Rank}

\subsection{Stability (A Normal Form)}

Throughout this paper, let $M$ denote an entry-wise nonnegative matrix of dimension $m \times n$. We will also let $M_i$ denote the $i^{th}$ column of $M$ and $M^j$ denote the $j^{th}$ row. Given a subset $U \subset [n]$, we will let $M_U$ denote the submatrix consisting of columns of $M$ from the set $U$ (and similarly $M^V$ is a submatrix of rows of $M$). 

\begin{definition}
$rank^+(M)$ is the smallest $r$ such that $M$ can be written as $$M = AW$$ where $A$ and $W$ are nonnegative and have dimension $m \times r$ and $r \times n$ respectively. 
\end{definition}

Additionally, we will call $M = AW$ a nonnegative matrix factorization of inner-dimension $r$. 

\begin{definition}
$$aff(A) = \Big \{\sum_i \alpha_i A_i \Big | \forall_i \alpha_i \geq 0\Big \}$$
\end{definition}

($aff(A)$ is the affine hull of columns in $A$). 

Note: Given $A$, there is a nonnegative matrix $W$ such that $M = AW$ if and only if each column $M_i$ of $M$ is contained in $aff(A)$. 

\begin{definition}
Given $A$ and a vector $v \in \mathbb{R}^m$ (recall $A$ is dimension $m \times r$), we will call a subset $S$ of columns of $A$ \textbf{admissible} if $$v \in aff(A_S).$$
\end{definition}

We will use this notion to place a stronger requirement on any nonnegative matrix factorization of $M$ It will not be immediately clear, but as we will be able to add this requirement without loss of generality. 

Throughout this paper, we will make use of the \textbf{lexicographic ordering} on subsets of columns of $A$. The standard lexicographic ordering is often restricted to comparing to subsets of the same size, but here we will want an ordering on all subsets. We will simply impose that if $|S| < |T|$, $S$ is before $T$ in the lexicographic ordering. 

Let $M = AW$ be a nonnegative matrix factorization.

\begin{definition}
For each column $M_i$, let $S_i$ be the lexicographically first admissible subset (of columns of $A$) for $M_i$. Similarly, for each row $M^j$, let $T_j$ be the lexicographically first admissible subset (of rows of $W$) for $M^j$. We call $M = AW$ \textbf{stable} if:
\begin{enumerate}

\item for each $i$, $W_i$ is supported in $S_i$

\item and for each $j$, $A^j$ is supported in $T_j$. 

\end{enumerate}
\end{definition}

Next we show that a nonnegative matrix factorization of inner-dimension $r$ can always be made stable (while preserving nonnegativity and the inner-dimension):

\begin{lemma}~\label{lemma:stable}
If $M = AW$ is a nonnegative matrix factorization of inner-dimension $r$, then there is a $\tilde{A}$ and $\tilde{W}$ such that:
\begin{enumerate}

\item $M = \tilde{A} \tilde{W}$, $\tilde{A}$ and $\tilde{W}$ are nonnegative and have inner-dimension $r$ and

\item $M = \tilde{A} \tilde{W}$ is stable.

\end{enumerate}
\end{lemma}

\begin{proof}
The natural approach to prove this lemma is, if $M = AW$ is not stable, update columns in $W$ or rows in $A$. The only subtle point is that if we update $A$ and $W$ \textbf{at the same time} to $\tilde{A}$ and $\tilde{W}$, we may not have $M = \tilde{A} \tilde{W}$. So the approach is to update only one of these two at a time, to preserve that $M = A \tilde{W}$ or $M = \tilde{A} W$ and \textbf{then} update the other. Suppose we update $W$ to $\tilde{W}$ first. Then for a row in $M^j$, the lexicographically first subset of admissible rows (for $M^j$) is defined with respect to $\tilde{W}$ and not $W$ - i.e. a subset $V$ of rows is admissible if $M^j \in aff(\tilde{W}^V)$. 

Throughout our updating process, we will make use of a potential function to ensure that this process terminates. To each row of $A$ and to each column of $W$, we will associate a subset of $[r]$ corresponding to the support of the vector. Whenever we update either a row of $A$ or a column of $W$, the support will only ever move \textbf{earlier} according to the lexicographic ordering. 

So, now we can define our updating procedure. We start with $M = AW$, and each update phase will alternately be an $A$-updating phase or a $W$-updating phase. In a $W$-updating phase, for each column $M_i$ let $S_i$ be the lexicographically first subset of columns of $A$ that is admissible for $M_i$. If $S_i$ is lexicographically (strictly) earlier than the support of $W_i$, we find a vector $\tilde{W}_i$ that is nonnegative, and supported in $S_i$ and satisfies $M_i = A \tilde{W}_i.$ If not, we set $\tilde{W}_i = W_i$. In either case, we have that $M_i = A \tilde{W}_i$ and hence $M = A \tilde{W}$. At the end of this phase, we overwrite $W$ with $\tilde{W}$. 

The $A$-updating phase is defined analogously, and throughout this procedure we maintain the invariant that $M = AW$ and $A$ and $W$ are nonnegative and have inner-dimension $r$. Note that the support of columns of $W$ and rows of $A$ are monotonically decreasing according to the lexicographical ordering, and if either $A$ or $W$ are updated (any row of $A$ or any column of $W$), one support must have strictly decreased according to the lexicographic ordering so this updating procedure terminates with $M = \tilde{A} \tilde{W}$ that are nonnegative, have inner-dimension $r$ and are also stable. 
\end{proof}

\subsection{Few Entries Determine $A$ and $W$}

Throughout this section, let $M = AW$ be a \textbf{stable} nonnegative matrix factorization. 

The goal in this section is to demonstrate that (given $M$), only a few entries in $A$ and $W$ are needed to determine the remaining entries. This is only a property of stable factorizations, and is not guaranteed to hold for general factorizations. 

Let $rank(A) = s$ and let $U \subset [m]$ be a set of $s$ linearly independent rows in $A$. Furthermore, let $S_1, S_2, ... S_p \subset [r]$ be the (full) list of sets of $s$ linearly independent columns of $A$ (in lexicographic order). Note that $p \leq {r \choose s} \leq 2^r$. 

\begin{definition}
The ensemble of $A$ (at $U$) is a list of linear transformations: $B_1, B_2... B_p$ where for each $i$, $B_{i}$ is an $r \times s$ matrix that is zero on all rows outside the set $S_{i}$ and restricted to rows in $S_i$ is $(A^U_{S_i})^{-1}$. 
\end{definition}

Note that each submatrix $(A^U_{S_i})^{-1}$ is indeed invertible: $rank(A) = s$ and $U$ is a set of $s$ linearly independent rows so a set $S_i$ of columns of $A$ is linearly independent if and only if these vectors restricted to $U$ are also linearly independent. 

The main goal in this section is to show:

\begin{lemma}\label{lemma:lex}
For each column $M_i$, among the set of vectors $$\calS = \Big \{ B_1 M_i^U, B_2 M_i^U, ... B_p M_i^U \Big \}$$ $W_i$ is the \textbf{unique} vector with lexicographically minimal support among all nonnegative vectors in the set $\calS$.
\end{lemma}

We will break this lemma up into two parts:

\begin{lemma}\label{lemma:occur}
$W_i$ is contained in the set $\calS$.
\end{lemma}

\begin{proof}
Let $R_i$ be the support of $W_i$. Then $R_i$ must correspond to a linearly independent set of columns of $A$ -- otherwise we could find a nonnegative $\tilde{W}_i$ whose support is a strict subset of $R_i$ such that $A \tilde{W}_i = M_i$, but this would violate the condition of stability.

Because the sets of linearly independent columns of $A$ are a matroid, there is a set $S_{i'}$ of $s$ linearly independent columns of $A$ for which $R_i \subset S_{i'}$. Hence $$B_{i'} M_i^U = B_{i'}(A W_i)^U = B_{i'}A^U W_i = v.$$ However, $B_{i'}$ is zero on rows outside the set $S_{i'}$ and restricting $B_{i'}A^U$ to rows and columns in $S_{i'}$ is the $s \times s$ identity matrix. Since the support of $W_i$ is contained in $S_{i'}$, we have $W_i = v$.
\end{proof}

We note a corollary of this lemma that will be useful later:

\begin{corollary}\label{cor:indep}
The support of $W_i$ corresponds to a linearly independent set of columns in $A$.
\end{corollary}

Next, we prove the second part needed for the main result in this section:

\begin{lemma}\label{lemma:isvalid}
For each vector $B_{i'}M_i^U$, $A B_{i'} M_i^U = M_i$.
\end{lemma}

\begin{proof}
Let $v = A B_{i'} M_i^U$. We prove this lemma in two parts: first we prove that $v^U = M_i^U$ and then we prove the full lemma from this. Since $B_{i'}$ is zero on rows outside the set $S_{i'}$, we have $$A B_{i'} = A_{S_{i'}} B_{i'}^{S_{i'}} = A_{S_{i'}}  (A_{S_{i'}}^U)^{-1}.$$ Hence $v^U = A_{S_{i'}}^U  (A_{S_{i'}}^U)^{-1} M_i^U = M_i^U$. 

Consider a $j$ outside the set $U$. By the choice of $U$, the row $A^j$ can be expressed as a linear combination of rows in $A$ in the set $U$: $$A^j = \sum_{j' \in U} \alpha_{j, j'} A^{j'}$$ Since $A W_i = M_i$, we have $ M_i^j = A^j W_i = \sum_{j' \in U} \alpha_{j, j'} A^{j'} W_i = \sum_{j' \in U} \alpha_{j, j'} M_i^{j'}$ and hence:
\begin{eqnarray*}
v^j &=& A^j B_{i'} M_i^U = \sum_{j' \in U} \alpha_{j, j'} A^{j'} B_{i'} M_i^U \\
&=& \sum_{j' \in U} \alpha_{j, j'} v^{j'} = \sum_{j' \in U} \alpha_{j, j'} M_i^{j'} = M_i^j
\end{eqnarray*}
\end{proof}

Now we can prove the main lemma in this section:

\vspace{0.5pc}

\begin{proof}
We have already shown (Lemma~\ref{lemma:occur}) that $W_i$ occurs in the set $\calS$. Consider any other \textbf{nonnegative} vector $B_{i'} M_i^U = v$. We need to show that the support of $v$ is lexicographically later than the support of $W_i$. 

First, we claim that if $v \neq W_i$ then the support of $W_i$ is not the same as the support of $v$. Suppose not - i.e. $v \neq W_i$ and yet the support of $v$ and of $W_i$ are identical (let this set be $R$). Indeed $R$ must correspond to a linearly independent set of columns of $A$ (Corollary~\ref{cor:indep}). Hence we cannot have $A(v - W_i) = \vec{0}$ (using Lemma~\ref{lemma:isvalid}) with $v - W_i \neq \vec{0}$ and support of $v - W_i$ contained in $R$. 

So the support of $W_i$ and $v$ are not identical and one of these must be lexicographically earlier. Suppose (for contradiction) that the support of $v$ is earlier. We know (Lemma~\ref{lemma:occur}) that the support of $W_i$ is an admissible set of columns of $A$ for $M_i$. This contradicts stability (because we could update $W_i$ to $v$), and so we can conclude that the support of $W_i$ is lexicographically earlier. 
\end{proof}

Let $rank(W) = t$ and let $V$ be a set of $t$ linearly independent columns of $W$. Then we can define an ensemble $C_1, C_2, ... C_q$ for $W$ at $V$ analogously as we did for $A$. Similarly, we have $q \leq {r \choose t}$ and for all $j$, among the set $$\calT = \Big \{ M_V^j C_1, M_V^j C_2, ... M_V^j C_q \Big \}$$ $A^j$ is the vector with lexicographically minimal support among all nonnegative vectors in $\calT$ (this follows from the above proof by interchanging the roles of $A$ and $W$). 

\subsection{A Semi-Algebraic Set, Take 1}

Our goal is to encode the question of whether or not $rank^+(M) \leq r$ as a non-emptiness problem for a semi-algebraic set with a small number of variables. Our first attempt will be to choose the entries in $B_1, B_2, ... B_p$ and $C_1, C_2, ... C_q$ as the \textbf{variables}. Our first goal is to construct a set of polynomial constraints (using the variables) so that setting $B_1, B_2, ... B_p$ and $C_1, C_2, ... C_q$ to the ensembles of a stable factorization $M = AW$ is a valid solution. We then show (conversely) that any valid setting of the variables in fact yields a nonnegative matrix factorization with inner-dimension $r$. 

Suppose we are given the sets $U$ and $V$, and the ensembles $B_1, B_2, ... B_p$ and $C_1, C_2, ... C_q$.

\begin{definition}
Let $first(\calS)$ applied to a collection of vectors output the vector with lexicographically minimal support among all nonnegative vectors in $\calS$. 
\end{definition}

This function can output \textbf{FAIL} if there is no nonnegative vector in $\calS$. 

\begin{claim}
Set:
\begin{equation} \label{eq:setw}
W_i \leftarrow first(\{B_1 M_i^U, B_2 M_i^U, ... B_p M_i^U\})
\end{equation}
\begin{equation} \label{eq:seta}
A^j \leftarrow first(\{M^j_V C_1, M^j_V C_2, ... M^j_V C_q\})
\end{equation}

\noindent There is an explicit Boolean function $\mathbb{P}$ that determines if for all $i$ and $j$: \textbf{1.} $ W_i \geq \vec{0}$  \textbf{2.} $A^j \geq \vec{0}$   and \textbf{3.} $ A^j W_i = M^j_i. $
Furthermore, $\mathbb{P}$ is a function of sign constraints on the polynomials:
\begin{enumerate}

\item $B_{i'}M_i^U$ (for all $i, i'$)

\item $M^j_V C_{j'}$ (for all $j, j'$) and

\item $M^j_V C_{j'} B_{i'} M_i^U - M_i^j$ (for all $i, i', j, j'$).

\end{enumerate}
\end{claim}

\vspace{0.5pc}

This claim is immediate, but we include a description of the Boolean function $\mathbb{P}$ for completeness

\vspace{0.5pc}

\begin{proof}
The Boolean function $\mathbb{P}$ will be an AND over subfunctions $\mathbb{P}_{i,j}$ defined for each $i$ and $j$: $\mathbb{P}_{i,j}$ will compute the index $i'$ and $j'$ so that $B_{i'}M_i^U$ and $M^j_V C_{j'}$ are lexicographically earliest among nonnegative vectors in the sets $\calS = \{B_1 M_i^U, B_2 M_i^U, ... B_p M_i^U\}$ and $\calT = \{M^j_V C_1, M^j_V C_2, ... M^j_V C_q\}$ respectively. This can be computed from only the signs of entries in the vectors in these sets. 

Then $\mathbb{P}_{i,j}$ will check that for this $i'$ and $j'$, that $M^j_V C_{j'} B_{i'} M_i^U = M_i^j$. If there is no nonnegative vector in either $\calS$ or $\calT$, or there are two or more nonnegative vectors tied for lexicographically earliest support (among only nonnegative vectors) then $\mathbb{P}_{i,j}$ will output \textbf{FAIL}. 
\end{proof}

\begin{lemma}\label{lemma:complete}
$\mathbb{P}$ will output \textbf{PASS} when $\{B_{i'}\}_{i'}$ and $\{C_{j'}\}_{j'}$ are chosen as the ensembles of a stable factorization $M = AW$.
\end{lemma}

\begin{proof}
This follows immediately from Lemma~\ref{lemma:lex}. However, note that Lemma~\ref{lemma:lex} establishes uniqueness (i.e. the vector with lexicographically earliest support among all nonnegative vectors is unique) and hence each $\mathbb{P}_{i,j}$ will not prematurely output \textbf{FAIL} for these choices of $\{B_{i'}\}_{i'}$ and $\{C_{j'}\}_{j'}$. 
\end{proof}

Next, we prove the converse direction:

\begin{lemma}\label{lemma:sound}
If $\mathbb{P}$ outputs \textbf{PASS}, then $A$ and $W$ (as defined in \ref{eq:setw} and \ref{eq:seta}) are a nonnegative matrix factorization of inner-dimension $r$. 
\end{lemma}

Note that this factorization is not necessarily stable. 

\vspace{0.5pc}

\begin{proof}
We have that $W_i$ and $A^j$ are nonnegative (otherwise $\mathbb{P}$ would have output \textbf{FAIL}) and $\mathbb{P}$ explicitly checks that $A^j W_i = M_i^j$ and hence $M = AW$. Note that $B_{i'}$ and $C_{j'}$ are $r \times s$ and $t \times r$ dimensional, so $M = AW$ does indeed have inner-dimension $r$. 
\end{proof}

Combining Lemma~\ref{lemma:complete} and Lemma~\ref{lemma:sound}, we have

\begin{theorem}
$\mathbb{P}$ outputs \textbf{PASS} for some choice of $s, t, U, V, p$ and $q$ and some setting of the variables $B_1, B_2, ... B_p$ and $C_1, C_2, .. C_q$  if and only if $rank^+(M) \leq r$.
\end{theorem}

This leads to a natural approach for computing the nonnegative rank:

\begin{enumerate} \itemsep 0pt
\small
\tt
\item Guess $s = rank(A)$, $t = rank(W)$ (for some stable factorization $M = AW$)

\item Guess $U$ and $V$

\item Guess $p \leq {r \choose s}$ and $q \leq {r \choose t}$

\item Define a semi-algebraic set where the entries of $B_1, B_2, ... B_p$ and $C_1, C_2, ... C_q$ 
are variables (using the Boolean function $\mathbb{P}$)

\item Run an algorithm for deciding if the semi-algebraic set is non-empty (e.g. \cite{Ren})

\end{enumerate}

The running-time of the best algorithms for deciding if a semi-algebraic set is non-empty run in time $$ \Big ( \mbox{$\#$ polynomials} \times D \Big )^{O(k)}$$ where $D$ is the maximum degree and $k$ is the number of \textbf{variables}. This bound is close to (optimal) bounds on the number of sign configurations of  a set of polynomials with maximum degree $D$ and $k$ variables. These bounds are due to a number of authors, but are often referred to as Milnor-Warren bounds. Indeed the main bottleneck in algorithms for determining non-emptiness for a semi-algebraic set is just the time needed to enumerate all of these sign configurations (and make an oracle call to the Boolean function for each one). 

In the approach above, there are $r(p + q) + mnpq$ polynomials of degree at most \textbf{two} in the variables. $r(p+q)$ constraints are due to nonnegativity and $mnpq$ constraints are used to ensure that $M = AW$. However, the drawback of the above approach is that the number of \textbf{variables} is large. 

There are $rsp + rtq$ variables, and indeed $p$ and $q$ can be exponential in $r$. For example, if we take the columns of $A$ to be vertices of the cross-polytope (in $r/2$ dimensions), then we do in fact need exponentially many simplices (one corresponding to each linear transformation $B_{i'}$) to cover the convex hull of the cross-polytope just by a facet-counting argument. 

Hence, the running time of the above algorithm will be doubly exponential in $r$. However, we will be able to reduce the number of variables in this semi-algebraic set to polynomial in $r$ (and we emphasize that this is possible only for the semi-algebraic set we defined here, not for the semi-algebraic set define in Arora et al \cite{AGKM}). The definition of stability is somewhat delicate, but this is what allows us to get an exponential reduction in the number of variables. 

\subsection{A Semi-Algebraic Set, Take 2}\label{sec:take2}

Here we reduce the number of variables in the semi-algebraic set \textbf{exponentially} by exploiting algebraic dependence among the matrices in the ensembles. 

Consider the ensemble: $B_1, B_2, ... B_p$ where for each $i$, there is a linearly independent set $S_i$ of $s$ columns of $A$ and $(B_i)^{S_i} = (A^U_{S_i})^{-1}$. 
Recall Cramer's Rule:

\begin{lemma} [Cramer]
Let $R$ be an $s \times s$ invertible matrix. Then $(R^{-1})_i^j =  det(R^{-i}_{-j}) / det(R)$
where $R^{-i}_{-j}$ is the matrix $R$ with the $i^{th}$ row and the $j^{th}$ column removed. 
\end{lemma}

\noindent Hence we can instead use a variable for each entry in $A^U$ and each entry in $W_V$. Then the sign of $(A^U_{S_{i'}})^{-1} M_i^U$ can be recovered as a Boolean function of signs of degree at most $s^2$ polynomials in the entries of $A^U$. (Additionally, we can check whether or not the polynomial $det(A_{S_{i'}}^U)$ is non-zero to determine if $S_{i'}$ is linearly independent).

Similarly a constraint of the form $$\sum_{\ell \in T_{j'} \cap S_{i'}} \Big (M_V^j (W_V^{T_{j'}})^{-1} \Big )_{\ell} \Big ( (A^U_{S_{i'}})^{-1} M_i^U \Big )^{\ell} = M^j_i$$ can be written as a degree at most $2r^2$ polynomial constraint in the entries of $A^U$ and $W_V$ by clearing the denominators by $det(W_V^{T_{j'}})$ and $det(A_{S_{i'}}^U)$. 

This new semi-algebraic set has $rs + rt$ variables and has $r(p+q) + (p+q) + mnpq$ polynomials of degree at most $2r^2$ (where the additional polynomials are the denominators in Cramer's Rule). 

Note that Lemma~\ref{lemma:sound} still implies that if $\mathbb{P}$ outputs \textbf{PASS}, $rank^+(M) \leq r$ and a nonnegative matrix factorization of inner-dimension $r$ can be computed from the settings of the variables for the valid point in the semi-algebraic set. And Lemma~\ref{lemma:lex} still implies that this semi-algebraic set is non-empty if $rank^+(M) \leq r$ (since moreover Lemma~\ref{lemma:stable} implies that there is a stable factorization). 

We can now use this reduction -- and known algorithms for solving systems of polynomial inequalities (as described in Section~\ref{sec:intro2}) to give a nearly optimal algorithm for deciding if $M$ has nonnegative rank at most $r$. Additionally, if $rank^+(M) \leq r$ we can also compute the corresponding nonnegative factors $A$ and $W$ to within an additive $\delta$ (at the expense of an extra factor $\log \frac{1}{\delta}$ in the running time). In \cite{Ren}, Renegar gave the first algorithm for deciding if a system of polynomial inequalities has a solution that runs in time exponential in the number of variables. We note that in \cite{Ren2}, Renegar extended this algorithm to also return a $\delta$-approximate solution to an algebraic formulae, and this is the algorithm that we will use to actually compute the factors $A$ and $W$. We also note that these algorithms only assume access to an oracle to the Boolean function $\mathbb{P}$, and our function $\mathbb{P}$ is computable in polynomial time. 

Let $L$ denote the maximum bit complexity of any coefficient in $M$. Then applying the algorithms in \cite{Ren} and \cite{Ren2} with our reduction we obtain: 

% Using Renegar's Algorithm, we obtain:

% Let $L$ denote the maximum bit complexity of any coefficient in $M$. 

\begin{maintheorem*}
There is a $poly(n, m, L) ( r 4^{r+1} m n)^{c r^2} $ time algorithm for deciding if the nonnegative rank of $M$ is at most $r$. Additionally, given $\delta > 0$ (and if $rank^+(M) \leq r$), the algorithm runs in time $poly(n, m, L, \log \frac{1}{\delta}) ( r 4^{r+1} m n)^{c r^2} $ returns factors $\tilde{A}$ and $\tilde{W}$ that are entry-wise close (within an additive $\delta$) to $A$ and $W$ (respectively) that are a nonnegative matrix factorization of $M$ of inner-dimension at most $r$. Furthermore the entries of $\tilde{A}$ and $\tilde{W}$ have rational coordinates with numerators and denominators bounded in bit length by $O(L ( r 4^{r+1} m n)^{c r^2} + \log \frac{1}{\delta})$. 
\end{maintheorem*}

\noindent Alternatively, in the Blum-Shub-Smale (BSS) Model \cite{BCSS} one can instead use the algorithm in \cite{Ren} to decide if $rank^+(M) \leq r$ and the running time of this algorithm is $poly(n, m) + ( r 4^{r+1} m n)^{c r^2}$. 

We emphasize that the above algorithm is based on answering a purely algebraic question: How many variables are needed (in a system of polynomial inequalities) to encode the question does $M$ have nonnegative rank at most $r$? We obtain an exponential improvement on the number of variables, over the results in \cite{AGKM}, and this coupled with algorithms for computing a solution to a system of polynomial inequalities, has an immediate algorithmic implication. The algorithm we obtain here is in fact nearly optimal under the Exponential Time Hypothesis (ETH) of Impagliazzo and Paturi \cite{IP}, since Arora et al \cite{AGKM} showed that an algorithm that decides if $rank^+(M) \leq r$ in $(nm)^{o(r)}$ time would imply a sub-exponential time algorithm for $3$-SAT. {\em It is somewhat surprising that an algorithm for computing the nonnegative rank can be designed based on reasoning about systems of polynomial inequalities, and no algorithm (under plausible complexity assumptions) can do much better. }

\section{Fragile Instances of Nonnegative Rank}

An important property of the rank of a matrix is that if a given matrix $M$ has rank $r$, there is an $r \times r$ submatrix of $M$ that also has rank $r$. Hence, rank admits a small certificate that serves as proof that a matrix does indeed have rank at least $r$ and this fact plays a crucial role in many applications. 

Here, we give highly fragile instances of nonnegative rank: We give a (nonnegative) matrix $M$ of dimension $n \times n$ with $rank^+(M) = 4r$, yet for {\em any} submatrix $N$ of at most  $\frac{n}{3r}$ columns of $M$, $rank^+(M) \leq 3r$. To put this result in context, consider a system of {\em linear} inequalities in $d$ dimensions that is infeasible. A basic result in discrete geometry \cite{Mat} is that there is a subset of at most $d + 1$ of the linear inequalities that is infeasible. In Section~\ref{sec:take2}, we gave a system of {\em polynomial} inequalities in $2r^2$ dimensions that has a solution if and only if $rank^+(M) \leq r$. One might hope that this system is infeasible if and only if there is a small subset of the inequalities that alone is infeasible, and that this would yield a subset of (say) the columns of $M$ that "proves" that $rank^+(M) > r$. Yet this is not the case and systems of polynomial inequalities do not have the "Helly Property" \cite{Mat} (indeed their individual constraints do not necessarily correspond to convex regions). 

To give fragile instances of nonnegative rank, we will make use of a series of reductions of Vavasis \cite{Vav} and a particular gadget in Arora et al \cite{AGKM}. In fact, we make use of a crucial property of the reduction in \cite{Vav} from nonnegative rank to the intermediate simplex problem  -- in a sense, that rows of $M$ are mapped to points and columns of $M$ are mapped to constraints when reducing to the intermediate simplex problem. We will only be interested in the intermediate simplex problem in two dimensions:

\begin{definition}
An instance of the intermediate polygon problem is a polygon $P \subset \mathbb{R}^2$ and a set $S \subset P$ of $|S| = n$ points. The goal is to find a triangle $T$ with $S \subset T \subset P$ in which case, we call this a \textbf{YES} instance and otherwise we call it a \textbf{NO} instance. 
\end{definition}

Our goal is to construct an explicit instance of this problem that is  \textbf{NO} instance and yet restricting to {\em any} set $S' \subset S$ of at most $\frac{n}{3}$ points is a \textbf{YES} instance and we accomplish this latter task by noticing that a particular gadget used in \cite{AGKM} (with a slight modification) has exactly this property. We will then be able to use this instance of the intermediate simplex problem as a gadget to construct fragile instances of nonnegative rank.

We will begin with some simple geometric lemmas and definitions. 

\begin{definition}
Let $C_d = \{ (x, y) | x^2 + y^2 \leq d\}$, and we will write $C$ for $C_1$. Let $o$ denote the origin. 
\end{definition}

\begin{definition}
Let $E$ be the set of all equilateral triangles $T \subset C$ where the vertices of $T$ are on the boundary of $C$. 
\end{definition}

In our arguments, we will also make use of the (largest) inner circle $c$ that is contained in all triangles in $E$. Equivalently, this circle is the intersection of all triangles in $E$:

\begin{definition}
Let $c = \cap_{T \in E} T = C_d$ where $d$ is defined as: (for an arbitrary $T \in E$), $d$ is the minimum distance from the boundary of $T$ to the origin. 
\end{definition}

Our instance of the intermediate polygon problem will be an intersection of $n$ triangles $T$ each in the set $E$. The common intersection of these triangles will contain $c$, and next we prove that in fact any triangle (contained in $C$) that contains $c$ must in fact be equilateral. This will help us reason about what sorts of triangles can make our instance a \textbf{YES} instance. The following two lemmas are proved in \cite{AGKM},  but we include the proofs here for completeness.

\begin{lemma}\cite{AGKM}~\label{lemma:c} 
Any arbitrary triangle $T$ with $c \subset T \subset C$ must be in the set $E$.
\end{lemma}

\begin{proof}
Consider a triangle $T$ with $c \subset T \subset C$. Then let $e_1, e_2$ and $e_3$ be the three edges of $T$ and let $\theta_1, \theta_2$ and $\theta_3$ be the {\em viewing angle} from the origin $o$, namely $\theta_i$ is the angle formed by $\langle a_i, o, b_i$ where $a_i$ and $b_i$ are the endpoints of $e_i$. 

Since $o \in T$, we have that $\theta_1 + \theta_2 + \theta_3 = 2 \pi$. Consider an edge $e_i$. We will prove, by contradiction, that $e_i \cap c$ must contain exactly one point (i.e. $e_i $ must be tangent to the circle $c$). Suppose not - since $c \subset T$, we must have that $e_i \cap c = \empty$. Then let $\ell$ be the line parallel to $e_i$ that is tangent to $c$. Let $e'_i$ be the intersection of $\ell$ with $C$. The viewing angle $\theta'_i$ of $e'_i$ is strictly larger than $\theta_i$, yet the intersection of any line $\ell$ tangent to $c$ with $C$ has viewing angle exactly $\frac{2\pi}{3}$ and hence we conclude that $\theta_1 + \theta_2 + \theta_3 < 2 \pi$, which is a contradiction.

So each $e_i$ is tangent to $c$ and in fact we can use a similar argument to conclude that each $e_i$ must be exactly the intersection of a line $\ell$ tangent to $c$ with $C$ (otherwise, again we would have that $\theta_1 + \theta_2 + \theta_3 < 2 \pi$). 

Hence, we conclude that each edge of $T$ has the same length, and each endpoint is on the boundary of $C$ so $T \in E$. 
\end{proof}

Throughout the remainder of this section, consider any finite set $T_1, T_2, ... T_n \in E$ of equilateral triangles, and let $S$ be the vertices of $\cap_{i = 1}^n T_i$. 

\begin{lemma} \cite{AGKM} ~\label{lemma:ti}
Let $T$ be a triangle with $S \subset T \subset C$. Then $T \in \{T_1, T_2, ... T_n\}$. 
\end{lemma}

\begin{proof}
Clearly we have that $conv(S) \subset T$ since $T$ is convex, and we also have that $c \subset conv(S) = \cap_{i = 1}^n T_i$. So by Lemma~\ref{lemma:c}, we can conclude that $T$ must be in $E$. Suppose that $T \notin  \{T_1, T_2, ... T_n\}$. 

Let $\{p_1, p_2, p_3\} = T \cap c$ (i.e. these are the three points on the boundary of $T$ closest to the origin). Similarly, for each $T_i$ let $\{p_1^i, p_2^i, p_3^i\} = T_i \cap c$. Then $\{p_1, p_2, p_3\}$ is a rotation (by $< \frac{2\pi}{3}$) of $\{p_1^i, p_2^i, p_3^i\} $ and hence $\{p_1, p_2, p_3\}$ are each strictly in the interior of $T_i$. 

Hence, $\{p_1, p_2, p_3\}$ are on the boundary of $conv(S) \cap T$ but not on the boundary of $conv(S)$, so $T$ cannot contain $conv(S)$. 
\end{proof}

\begin{lemma}~\label{lemma:bijection}
For each edge $e_j$ of a triangle $T_i$, $|e_j \cap S| = 2$ and furthermore for each $s \in S$, $s$ intersects the edges of exactly two (distinct) triangles in $\{T_1, T_2, ... T_n\}$. 
\end{lemma}

\begin{corollary}
$|S| = 3n$
\end{corollary}

\begin{proof}
Each edge of $conv(S)$ is by definition a subsegment of some unique edge $e_j$ of some triangle in $\{T_1, T_2, ... T_n\}$. All we need to show is that to each edge $e_j$ (of some triangle in $\{T_1, T_2, ... T_n\}$) we can find an edge of $conv(S)$ which is a subsegment of $e_j$:

Let $p_j$ be the closest point on $e_j$ to the origin. As we argued in Lemma~\ref{lemma:ti}, for all other triangles, $p_j$ is strictly in the interior. So the ray from the origin to $p_j$ hits the segment $e_j$ first (out of all edges of all triangles in $E$). Hence $p_j$ is on the boundary of $conv(S)$, but only one edge (namely $e_j$) contains $p_j$ so the edge of $conv(S)$ that contains $p_j$ is a subsegment of $e_j$, as desired. 
\end{proof}

As we noted, the gadget that we use here is a slight modification of the one in \cite{AGKM} -- and the modification that we need involves rescaling:

\begin{definition}
For each triangle $T \in E$, define $T^{(1-\epsilon)}$ as the scaling down of $T$ such that the vertices of $T^{(1-\epsilon)}$ are on the boundary of $C_{1 - \epsilon}$. 
\end{definition}

This rescaling is precisely what ensures that the original instance is a \textbf{NO} instance, but as we will see, if $\epsilon$ is sufficiently small then {\em every} small subset of $S$ is a \textbf{YES} instance. 

\begin{definition}
Let $S_i$ be the vertices of $conv(S) \cap T_i^{(1-\epsilon)}$. 
\end{definition}

\begin{claim}~\label{claim:scale}
If $\epsilon$ is sufficiently small, then $S_i = S - T_i \cap S$. 
\end{claim}

\begin{proof}
Recall that $conv(S) = \cap_{i = 1}^n T_i$. Consider an edge $e_j$ of $T_i$. Using Lemma~\ref{lemma:bijection}, $|e_j \cap S| = 2$, and we can choose $\epsilon$ small enough such that the region strictly between $e_j^{(1- \epsilon)}$ (namely, the corresponding edge in $T_i^{(1-\epsilon)}$) and $e_j$ does not contain any points in $S$, in which case $S_i = S - T_i \cap S$. 
\end{proof}

So consider the following instance of the intermediate polygon problem:
\begin{itemize}
\item Let $P = conv(\mbox{vertices in }T_i^{(1-\epsilon)})$ 
\item and let $S = \mbox{vertices of }\cap_{i = 1}^n T_i$. 
\end{itemize}

\begin{claim}
$(P, S)$ is a \textbf{NO} instance. 
\end{claim}

\begin{proof}
$P \subset C_{1 - \epsilon}$ by the definition of $T_i^{(1-\epsilon)}$, and using Lemma~\ref{lemma:c}, any triangle $T$ contained in $C$ that contains $S$ must be in the set $E$; and since any triangle in $E$ has its vertices on the boundary of $C$, we conclude that $T$ is not contained in $C_{1-\epsilon}$ and hence $(P, S)$ is indeed unsatisfiable. 
\end{proof}

\begin{lemma}
For any $S' \subset S$ with $|S'| < n$, $(P, S')$ is a \textbf{YES} instance. 
\end{lemma}

\begin{proof}
Using Lemma~\ref{lemma:bijection}, each $s \in S$ intersects exactly two edges of triangles in $\{T_1, T_2, ... T_n\}$, so if $|S'| < n$, there must be a triangle $T_i$ for which $T_i \cap S' = \empty$. 

Consider $T_i^{(1-\epsilon)}$: Using Claim~\ref{claim:scale}, we conclude that $T_i^{(1-\epsilon)} \cap S' = S' - T_i \cap S' = S'$. And we have that $T_i^{(1-\epsilon)} \subset C_{1 - \epsilon}$, so $(P, S')$ is indeed satisfiable. 
\end{proof}

We use the following lemma from Vavasis:

\begin{lemma}\cite{Vav}
Let $rank(M) = r$, and let $M = UV$ where $U$ and $V$ have $r$ columns and rows respectively. Then $M$ has $rank^+(M) = r$ if and only if there is an invertible $r \times r$ matrix $Q$ such that $U Q^{-1}$ and $Q V$ are both nonnegative. 
\end{lemma}

We could use the reduction in \cite{Vav} from nonnegative rank to the intermediate simplex problem, but there is a technical issue that arises. Here, we give a slight modification of this reduction that avoids this issue:

Consider the plane $F = \{(x, y, z) | x + y + z = 1\} $. Map $P$ to this plane so that $P$ is contained in the nonnegative orthant (scale down $P$, if need be), and let the nonnegative hull of vectors in $P$ and the origin be denoted by the cone $\calC$. 

Let $\calC = \{ \vec{v} | A v \geq 0\}$ and set the rows of $U$ to be vertices of $F \cap \calC$ and let $V = A^T$. Note that the vertices of $F \cap \calC$ are just the {\em three-dimensional} coordinates corresponding to the points in $S$. Note that $UV$ is a nonnegative matrix, since each vertex of $F \cap \calC$ is contained in the cone $\calC$. This reduction is essentially the one in \cite{Vav} but with a minor change to avoid a certain technical issue that would arise otherwise. 

\begin{lemma}
There is an invertible $r \times r$ matrix $Q$ such that $U Q^{-1}$ and $Q V$ are both nonnegative if and only if $(P, S)$ is a \textbf{YES} instance.
\end{lemma}

\begin{proof}
Suppose $(P, S)$ is a \textbf{YES} instance. Let the rows of $Q$ be the {\em three-dimensional} coordinates of the vertices of the triangle $T$ (i.e. these are the vectors on the plane $F$). These points are in the cone $\calC$, so $QV$ is nonnegative. Furthermore, $S \subset T$ so each row of $U$ is in the convex hull of rows of $Q$ and $UQ^{-1}$ is nonnegative. 

Conversely, consider an invertible $Q$ for which $U Q^{-1}$ and $Q V$ are both nonnegative. For each row in $Q$, let $p_i$ be the intersection of the ray through the origin and the row in $Q$ with $F$. $p_i \in \calC$, so the associated {\em two-dimensional} point is in $P$. Furthermore, each row of $U$ is in the nonnegative hull of $\{p_1, p_2, p_3\}$ and each $p_i$ and each row in $U$ has nonnegative entries and the sum of the entries is one. Hence each $p_i$ and each row in $U$ has unit $\ell_1$ norm. So each row of $U$ is in the convex hull of $\{p_1, p_2, p_3\}$, and so the associated {\em two-dimensional} triangle contains $S$. 
\end{proof}

Note that in this reduction, rows of $M = UV$ are mapped one-to-one to points in $S$ and columns of $M$ are mapped one-to-one to facets in $P$. Hence, $(U, V)$ is a \textbf{NO} instance, but any set of $< n$ rows of $U$ is a \textbf{YES} instance. 

So $M = UV$ is a nonnegative matrix of dimension $3n \times 3n$ with nonnegative rank $\geq 4$ and yet any submatrix of $< n$ rows has nonnegative rank $\leq 3$. We can use this matrix $M$ to construct a $3rn \times 3rn$ matrix which is block diagonal, and has $M$ along the diagonal. Then:

\begin{maintheorem*}
For any $r \in \mathbb{N}$, there is a $3rn \times 3rn$ nonnegative matrix which has nonnegative rank at least $4r$ and yet for any $< n$ rows, the corresponding submatrix has nonnegative rank at most $3r$. 
\end{maintheorem*}

\noindent An interesting open question is to characterize the family of matrices for which nonnegative rank {\em can} be certified by a small submatrix, since in many applications is is quite natural to assume that the input matrices satisfy these conditions. 

\section*{Acknowledgements}

We would like to thank Jim Renegar for useful discussions. We would also like to thank Sanjeev Arora, Rong Ge, Pavel Hrubes and Avi Wigderson for helpful comments at a preliminary stage of this work.


\newpage

\begin{thebibliography}{99}

\bibitem{AGKM}
S. Arora, R. Ge, R. Kannan and A. Moitra.
\newblock Computing a nonnegative matrix factorization \-- provably.
\newblock {\em STOC} 2012, to appear.

\bibitem{AUY}
A. Aho, J. Ullman and M. Yannakakis.
\newblock On notions of information transfer in VLSI circuits.
\newblock {\em STOC}, pp. 133--139, 1983. 

\bibitem{BPR}
S. Basu, R. Pollack and M. Roy.
\newblock On the combinatorial and algebraic complexity of quantifier elimination.
\newblock {\em Journal of the ACM}, pp. 1002--1045, 1996. Preliminary version in {\em FOCS} 1994.

\bibitem{BCSS}
L. Blum, F. Cucker, M. Shub and S. Smale. 
\newblock {\em Complexity of Real Computations}.
\newblock Springer Verlag, 1998.  

\bibitem{CR93}
J. Cohen and U. Rothblum.
\newblock Nonnegative ranks, decompositions and factorizations of nonnegative matices.
\newblock {\em Linear Algebra and its Applications}, pp. 149--168, 1993.

\bibitem{FMPTW}
S. Fiorini, S. Massar, S. Pokutta, H. Tiwary and R. de Wolf.
\newblock Linear vs semidefinite extended formulations: exponential separations and strong lower bounds.
\newblock {\em STOC} 2012, to appear.

\bibitem{FRT}
S. Fiorini, T. Rothvo\ss $\mbox{ }$and H. Tiwary.
\newblock Extended formulations for polygons.
\newblock Arxiv, 2011.

\bibitem{GV}
G. Golub and C. van Loan.
\newblock {\em Matrix Computations}
\newblock The Johns Hopkins University Press, 1996.

\bibitem{IP}
R. Impagliazzo and R. Paturi. 
\newblock On the complexity of k-SAT. 
\newblock {\em JCSS} pp. 367--375, 2001.

\bibitem{LOG}
L. Lov\'asz and M. Saks. 
\newblock Communication complexity and combinatorial lattice theory.
\newblock {\em JCSS}, pp. 322--349, 1993. Preliminary version in {\em FOCS} 1988. 

\bibitem{Mat}
J. Matousek.
\newblock {\em Lectures on Discrete Geometry}.
\newblock Springer, 2002. 

\bibitem{Nisan}
N. Nisan.
\newblock Lower bounds for non-commutative computation (extended abstract).
\newblock {\em STOC}, pp. 410--418, 1991. 

\bibitem{Ren}
J. Renegar.
\newblock On the computational complexity and geometry of the first-order theory of the reals. 
\newblock {Journal of Symbolic Computation}, pp. 255-352, 1992. 

\bibitem{Ren2}
J. Renegar.
\newblock On the computational complexity of approximating solutions for real algebraic formulae.
\newblock {SIAM Journal on Computing}, pp. 1008--1025,1992. 

\bibitem{Sei}
A. Seidenberg.
\newblock A new decision method for elementary algebra.
\newblock {\em Annals of Math}, pp. 365--374, 1954. 

\bibitem{Tar}
A. Tarski.
\newblock A decision method for elementary algebra and geometry.
\newblock {\em University of California Press}, 1951. 

\bibitem{Vav}
S. Vavasis.
\newblock On the complexity of nonnegative matrix factorization.
\newblock {\em SIAM Journal on Optimization}, pp. 1364-1377, 2009.

\bibitem{Yan}
M. Yannakakis.
\newblock Expressing combinatorial optimization problems by linear programs.
\newblock {\em JCSS}, pp. 441--466, 1991. Preliminary version in {\em STOC} 1988.

\end{thebibliography}
\end{document}